\documentclass{llncs}

\usepackage{graphics}
\usepackage[dvips]{epsfig}

\usepackage{amsmath}
\usepackage{amssymb}
\usepackage{amsfonts}
\usepackage{graphicx}

\newcommand{\alg}{\hfill $\diamondsuit$}

\newcommand{\remove}[1]{}

\newcommand{\qq}{\hfill $\square$ \smallbreak}

\begin{document}

\title{{\bf Deterministic Broadcasting and Gossiping\\ with Beeps}}

\author{
Kokouvi Hounkanli
\and
Andrzej Pelc\thanks{Supported in part by NSERC discovery grant 8136 -- 2013
and by the Research Chair in Distributed Computing of
the Universit\'{e} du Qu\'{e}bec en Outaouais.}
}

\institute{D\'{e}partement d'informatique, Universit\'{e} du Qu\'{e}bec en Outaouais,\\
Gatineau, Qu\'{e}bec J8X 3X7,
Canada. E-mails: {\tt houk06@uqo.ca}, {\tt  pelc@uqo.ca}.}

\date{ }

\maketitle

\begin{abstract}
Broadcasting and gossiping are fundamental communication tasks in networks. In broadcasting, 
one node of a network has a message that must be learned by all other nodes. In gossiping, every node has a (possibly different) message,
and all messages must be learned by all nodes. We study these well-researched tasks in a very weak  communication model, called the
 {\em beeping model}. Communication proceeds in synchronous rounds. In each round, a node can either listen, i.e., stay silent,
or beep, i.e., emit a signal. A node hears a beep in a round, if it listens in this round and if one or more adjacent nodes beep in this round. All nodes have different
labels from the set $\{0,\dots , L-1\}$. 

Our aim is to provide fast  deterministic algorithms for broadcasting and gossiping in the beeping model. Let 
$N$ be an upper bound on the size of the network and $D$ its diameter. Let $m$ be the size of the message in broadcasting,
and $M$ an upper bound on the size of all input messages in gossiping. For the task of broadcasting we give an algorithm working in time $O(D+m)$ for arbitrary networks,
which is optimal. For the task of gossiping we give an algorithm working in time $O(N(M+D\log L))$ for arbitrary networks.

\vspace{2ex}


\vspace{2ex}
At the time of writing this paper we were unaware of the paper\\ A. Czumaj, P. Davis, Communicating with Beeps,  arxiv:1505.06107 [cs.DC]
which contains the same results for broadcasting and a stronger upper bound for gossiping in a slightly different model.
\end{abstract}




\section{Introduction}
{\bf The background and the problem.}
Broadcasting and gossiping are fundamental communication tasks in networks. In broadcasting, 
one node of a network, called the {\em source}, has a message that must be learned by all other nodes. In gossiping, every node has a (possibly different) input message,
and all messages must be learned by all nodes. We study these well-researched tasks in a very weak communication model, called the
 {\em beeping model}. Communication proceeds in synchronous rounds. In each round, a node can either listen, i.e., stay silent,
or beep, i.e., emit a signal. A node hears a beep in a round, if it listens in this round and if one or more adjacent nodes beep in this round. 
The beeping model has been introduced
in \cite{CK} for vertex coloring and used
in \cite{AABCHK} to solve the MIS problem. The beeping model is widely applicable, as it makes small demands on communicating devices, relying only on carrier sensing. In fact, as mentioned in 
\cite{CK}, beeps are an even weaker way of communicating than using one-bit messages, as the latter ones allow three different states (0,1 and no message),
while beeps permit to differentiate only between a signal and its absence. 

The network is modeled as a simple connected undirected graph. Initially all nodes are dormant. The adversary wakes up the source in the case of broadcasting
and some nonempty subset of nodes, at possibly different times, in the case of gossiping. A woken up node starts executing the algorithm.  A dormant node is woken up by a beep of any neighbor.
Our aim is to provide fast  deterministic algorithms for broadcasting and gossiping in the beeping model.
The time of broadcasting is defined as the number of rounds between the wakeup of the source and the round in which all nodes of the network acquire
the source message.  The time of gossiping is defined as the number of rounds between the wakeup of the first node and the round in which all nodes
acquire all messages. Messages are considered as binary strings and the size of a message is the length of this string.
In the case of broadcasting, our algorithm does not assume any information about the network, and it does not require any labeling of nodes.
In the case of gossiping, we assume that all nodes have different labels from the set $\{0,\dots , L-1\}$ and that they know $L$. Moreover, we assume that all nodes
know the same upper bound $N$ on the size of the network and the same upper bound $M$ on the size of all input messages.  
Without loss of generality we may assume that $N \leq L$. Indeed, the parameter $L$ known to nodes is an upper bound on the size of the network, as all nodes have different labels.
Let $D$ be the diameter of the network,
initially unknown to the nodes.


{\bf Our results.} 
For the task of broadcasting we give an algorithm working in time $O(D+m)$ for arbitrary networks, where $m$ is the size of the source message.
This complexity is optimal. For the task of gossiping we give an algorithm working in time $O(N(M+D\log L))$ for arbitrary networks.

Due to space restrictions several proofs are moved to the Appendix.

{\bf Related work.}
Broadcasting and gossiping have been studied in various models for over four decades. Early work focused on the telephone model, where in each round communication
proceeds between pairs of nodes forming a matching, and nodes that communicate exchange all previously acquired information. Deterministic broadcasting in this model has been studied, e.g.,  in \cite{SCH} and deterministic gossiping in \cite{BS}. In \cite{FPRU} the authors studied randomized broadcasting.
In the telephone model studies focused on the time of the communication task and on the number of messages it uses. Early literature on communication in the telephone and related models is surveyed in \cite{FrLa,HHL}. Fault tolerant aspects of broadcasting and gossiping are surveyed in \cite{Pe}.

More recently, broadcasting and gossiping have been studied in the radio model. While radio networks model wireless communication, similarly as the beeping model, 
in radio networks nodes send entire messages of some bounded, or even unbounded size in a single round, which makes communication drastically different from that in the beeping model.
The focus in the literature on radio networks was usually on the time of communication.
Deterministic broadcasting in the radio model was studied, e.g., in \cite{CGR,KP} and deterministic gossiping in \cite{CGR}. Randomized broadcasting was studied in
\cite{KM} and randomized gossiping in \cite{CGR2}. The book \cite{Ko} is devoted to algorithmic aspects of communication in radio networks. 

Randomized leader election in the radio and in the beeping model was studied in \cite{GH}. Deterministic leader election in the beeping model was studied in 
\cite{FSW}. The authors showed an algorithm working in time $O(D\log n)$ in networks of diameter $D$ with labels polynomial in the size $n$ of the network.


\section{Broadcasting}

In this section we consider the simpler of our two communication tasks, that of broadcasting. Even for this easier task, the restrictions of the beeping model
require the solution of the basic problem of detecting the beginning and the end of the transmitted message. The naive idea would be to
adapt the method of beeping waves, used in \cite{GH} in a different context, and transmit a message by coding bit 1 by a beep and bit 0 by silence,
other nodes relaying these signals after getting them. However, in this coding there is no difference between message $(10)$ and message $(100)$ because both these
messages are coded by a single beep. Hence we need to reserve some sequence of beeps to mark the beginning and end of a message, and code bits
by some other sequences of beeps and silent rounds. One way of defining such a coding is the following. Consider the message $\mu=(a_1,\dots , a_m)$ that
has to be transmitted by some node $v$. Let $b$ denote a round in which $v$ beeps, and let $s$ denote a round in which $v$ is silent. The beginning and end of the message are marked by the sequence $(bb)$, bit $1$ is coded by $(bs)$, and bit $0$ is coded by $(sb)$. Hence message $\mu$ is transmitted as the sequence
$(c_1,\dots ,c_{2m+4})$, where\\
$c_i=b$, for $i \in \{1,2,2m+3,2m+4\}$, and\\
$c_{2j+1}=b$ and $c_{2j+2}=s$, if $a_j=1$, for $j=1,\dots , m$,\\
$c_{2j+1}=s$ and $c_{2j+2}=b$, if $a_j=0$, for $j=1,\dots , m$.

A node $w$ hearing the sequence $(c_1,\dots ,c_{2m+4})$ of beeps and silent rounds can correctly decode message $\mu$ as follows. Upon hearing two beeps
in the first two rounds it divides the successive rounds into segments of length 2 and records all beeps and silent rounds until a segment with two consecutive
beeps. Each segment between the two segments $(bb)$ is either of the from $(bs)$ or of the form $(sb)$. The node $w$ decodes each segment  $(bs)$ as 1
and each segment $(sb)$ as 0. In this way, message $\mu$ is correctly reconstructed.  
Note that the time of transmitting this message is $2m+4$ and hence linear in its size.

We will call the above sequence of beeps and silent rounds, chosen by a node $v$ for the message $\mu$, the {\em canonical sequence} for $\mu$
transmitted by node $v$. The way of reconstructing message $\mu$ by node $w$ is called the {\em canonical decoding}.
Using canonical sequences we formulate our broadcasting algorithm that works for arbitrary networks.

\pagebreak
\noindent
{\bf Algorithm Broadcast}

Let $t$ be the round in which the source is woken up. Given a source message $\mu$, the source transmits the canonical sequence for $\mu$
in rounds $t+3(j-1)$, for $j=1,2,\dots,2m+4$, and stops. More precisely, if $c_j=b$, then the source beeps in round $t+3(j-1)$, and  
if $c_j=s$, then the source is silent in this round. In all other rounds the source is silent. Every other node $w$ that is woken up by a beep in round $r$,
beeps in round $r+1$. If it hears a beep in a round $r+3j$, for some positive integer $j$, it beeps in round $r+3j+1$. In all other rounds node $w$ is silent. It divides all rounds $r+3j$,
for $j=0,1,,...$, into segments of length 2. After the second segment when it hears $(bb)$, the node decodes the source message using the canonical decoding
of the sequence received in rounds $r+3j$, for $j=0,1,,...$. Then it beeps in the next round and stops. 
\alg

\begin{theorem}\label{broad-trees}
Algorithm Broadcast is a correct broadcasting algorithm working in any network of diameter $D$ in time $O(D+m)$, where $m$ is the size of the source message. The time complexity of this algorithm is optimal.
\end{theorem}

\begin{proof}
Consider any network and define its $i$th layer as the set of nodes at distance $i$ from the source. Consider any node $w$ in the $i$th layer of this network, other than the source, and suppose that $w$ is woken up in round $r$. Hence nodes in layer $i-1$ beep only in rounds $r+3j$, nodes in layer $i$ beep only in rounds $r+3j+1$, and nodes in layer $i+1$ beep only in rounds $r+3j+2$,
for some integer $j$. Consequently, in rounds $r+3j$ the node $w$ hears a beep in exactly these rounds in which nodes from the layer $i-1$ beep. By induction on the distance
of $w$ from the source we get that $w$ gets 
correctly the canonical sequence for the source message and hence decodes the source message correctly. If the source is woken up in round $t$ and $w$ is at distance $i$ from the source, it starts receiving transmissions in round $t+i-1$
and stops in round $t+i-1+3(2m+4)-2$. Since $i \leq D$, it follows that the algorithm works in time $O(D+m)$. 

In order to prove that this time complexity is optimal, it is enough to show that every algorithm requires at least time $\Omega(D)$ and at least time $\Omega(m)$.
The first fact is immediate because the farthest node from the source must be at distance at least $D/2$ from it, and no signal can get to this node from the source faster than after $D/2$ rounds. The second fact holds even in the two-node network. Suppose that some broadcasting algorithm transmits every message of size $m$ from one node of this network (the source) to the other, in time $x<m$. The number of sequences of length $x$ with terms from the set $\{b,s\}$ is $2^x<2^m$. Since the number of possible source messages of size $m$ is $2^m$, it follows that for two distinct source messages the source must behave identically, and hence the input of the other node is identical. Hence the other node must decode the same message in both cases, which contradicts the correctness of the algorithm. 
\qq
\end{proof}

\section{Gossiping}

In this section we investigate the more complex task of gossiping. One way to accomplish this task is to have each node broadcast its input message.
However, in the highly contrived beeping model, periods of broadcasting by different nodes should be disjoint, otherwise messages, transmitted as
series of beeps and silent rounds, risk to become damaged, when a node receives simultaneously a beep being part of the transmission of one message
and should hear a silent round being part of the transmission of another message. The receiving node will then just hear the beep and the transmission of the message requiring this 
round to be silent becomes scrambled.

In order to broadcast in disjoint time intervals, nodes must establish an order between them and reserve the $i$th time interval to the broadcast of
the $i$th node in this order. This yields the following high-level idea of a gossiping algorithm. First we establish a procedure that finds the node with
the largest label. This is done in such a way that all nodes learn the largest label. 
(Notice that we cannot use the leader election algorithm from \cite{FSW} because this algorithm works only under an additional strong assumption that all nodes that are woken
by the adversary -- and not by hearing a beep -- are woken simultaneously in the first round of the algorithm execution.)
Next, using this elected leader, all nodes are synchronized: they
agree on a common round, and hence can simultaneously start the rest of the algorithm execution. Then the procedure of finding the largest-labeled node  is repeated at most $N$ times, where $N$ is an upper bound on the size of the network, each time the
currently found largest-labeled node withdrawing from the competition. In this way, after at most $N$ repetitions all nodes know the order between them,
and subsequently they broadcast their values in disjoint time intervals, in this order.

We assume that all nodes know the size $L$ of the label space. Let $\lambda=\lceil \log L \rceil$ and let $\ell _i$ be the binary representation of the integer $i<L$. We assume that all sequences $\ell _i$ are of length $\lambda$, the representations being padded by a prefix of 0's, if necessary. 
Hence the representation of a smaller integer is lexicographically smaller than the representation of a larger integer. 
We also assume that all nodes know a common upper bound $N$ on the number of nodes in the network.

Our first procedure finds the largest label among a set $S$ called the set of {\em participating nodes}. 

\smallbreak
\noindent
{\bf Procedure Find Max}

When a node is woken up in round $\tau$, either by the adversary or by a beep, it defines the following round numbers:
$t_j=\tau+j(4N+1)$, for $j=0,1,\dots,\lambda$. Then, for $j=1,\dots,\lambda$, the node defines the time interval $I_j=[t_j-2N,t_j+2N]$.
Note that these intervals are pairwise disjoint. 

The node beeps in round $\tau+1$. The rest of the procedure is divided into $\lambda$ stages, corresponding to time intervals $I_j$, for  $j=1,\dots,\lambda$. 

First assume that the node is participating.
Let $\ell=(c_1,\dots,c_{\lambda})$ be the binary representation of the node's label, of length $\lambda$.
In the beginning of the first stage the node is {\em active}.

If the node is still active at the beginning of the $j$th stage, then it behaves as follows.
If $c_j=0$, the node listens in all rounds of the time interval  $I_j$ until it hears a beep. If it does not hear any beep, it remains active and proceeds
to stage $j+1$. If it hears a beep for the first time in some round $t$, then it beeps in round $t+1$ and becomes non-active. If $c_j=1$,
the node listens until it hears a beep or until round $t_j$, whichever comes earlier. If it hears a beep in some round $t<t_j$, it beeps in round $t+1$,
listens till the end of time interval $I_j$ and remains active. Otherwise, it beeps in round $t_j$, listens till the end of time interval $I_j$ and remains active.

If the node is non-active at the beginning of the $j$th stage, then it listens in all rounds of the time interval  $I_j$ until it hears a beep.
If it does not hear any beep, it remains non-active and proceeds to stage $j+1$. If it hears a beep for the first time in some round $t$, then it beeps in round $t+1$, listens till the end of time interval $I_j$ and remains non-active.

A non-participating node is never active. It listens in all rounds of the time interval  $I_j$ until it hears a beep.
If it does not hear any beep, it remains non-participating and proceeds to stage $j+1$. If it hears a beep for the first time in some round $t$, then it beeps in round $t+1$,  listens till the end of the stage and remains non-participating.

At the end of stage $\lambda$, the (unique) participating node that remained active till the end of this stage is the node with the largest label among participating nodes.
\alg

\begin{lemma}\label{find-max}
At the end of the execution of procedure Find Max,  there is exactly one active participating node. This node has the largest label among participating nodes.
All nodes know the label of this node. Procedure Find Max takes time $O(N\log L)$.
\end{lemma}

The goal of our next procedure is synchronizing all processors. The procedure will be used upon completion of Procedure Find Max, and hence we assume that
the largest label is known to all nodes. Let $z$ be the node with the largest label. Upon completion of Procedure Synchronization, each node declares a specific
round to be {\em red}, and this round is the same for all nodes.

\smallbreak
\noindent
{\bf Procedure Synchronization} 

Each node other than $z$ has an integer variable $level$ initially set to 0. Let $\nu= \lceil \log N \rceil$. 
For every integer $0 \leq i <N$, let $bin(i)$ be the  binary representation of $i$ of length $\nu$, padded by a prefix of 0's, if necessary.
A string $bin(i)$ will be transmitted by nodes of the network, using the canonical sequence, as it was done for broadcasting in Section 2.
We briefly recall this coding. Let $b$ denote a round in which $v$ beeps, and let $s$ denote a round in which $v$ is silent. 
In $2\nu+4$ rounds, node $v$ transmits the following message $num(i)$ . The beginning and end of the message are marked by the sequence $(bb)$, every bit $1$
of $bin(i)$ is coded by $(bs)$, and every bit $0$
of $bin(i)$ is coded by~$(sb)$. 

At the beginning of the procedure node $z$ transmits $num(0)$ starting in round $t+1$, where $t$ is the round in which $z$ completed the execution of
Procedure Find Max. After completing this transmission, node $z$ waits till round $t+N(2\nu+4)$
and declares it to be {\em red}. Every node other than $z$ that is at $level$ 0 waits until it hears two consecutive beeps. Then it partitions the following rounds into consecutive segments of length 2, and decodes each segment of the form $(bs)$ as 1 and each segment of the form $(sb)$ as 0. As soon as it hears a segment
$\sigma$ consisting of two beeps, it considers the previously decoded string of bits as the binary representation of an integer $j$. It changes the value of its variable $level$ to $j+1$, transmits $num(j+1)$ in $2\nu+4$ rounds, starting in the round $r$ following the segment $\sigma$, then waits till  round $r+(N-j-1)(2\nu+4)$ and declares it
to be {\em red}.
\alg 

\begin{lemma}\label{synch}
All nodes declare the same round to be {\em red}. Procedure Synchronization takes time $O(N\log N)$.
\end{lemma} 

\begin{proof}
We prove the following invariant by induction on $d$.
\begin{enumerate}
\item
in time interval $[t+d(2\nu+4)+1, t+(d+1)(2\nu+4)]$ the only message transmitted is $num(d)$ and it is transmitted by all nodes at distance $d$ from node $z$
and only by these nodes;
\item
a node at distance $d$ from the node $z$ declares round $t+N(2\nu+4)$ as {\em red}.
 \end{enumerate} 
 
 The invariant is clearly satisfied for $d=0$.
Suppose that it holds for some $d \geq 0$.
The only nodes that hear the beeps transmitted in time interval $[t+d(2\nu+4)+1, t+(d+1)(2\nu+4)]$ are those at distance at most $d+1$ from node $z$. 
The only nodes among them that have value of $level$ 0 are nodes at distance exactly $d+1$ from node $z$. Since all the nodes at distance $d$ from $z$
transmit the same message $num(d)$ in this time interval, they all beep exactly in the same rounds of the interval. Hence the value of $d$ is correctly decoded
by all nodes at distance $d+1$ from node $z$. These nodes, and only these nodes, transmit $num(d+1)$ in the time interval $[t+(d+1)(2\nu+4)+1, t+(d+2)(2\nu+4)]$.
This proves the first part of the invariant. All these nodes set their value of $level$ to $d+1$ and declare as {\em red} the round $r+(N-d-1)(2\nu+4)$,
where $r$ is the last round of the preceding time interval, i.e., $r=  t+(d+1)(2\nu+4)$. Hence the declared round is
$ t+(d+1)(2\nu+4) + (N-d-1)(2\nu+4)=t+N(2\nu+4)$, which proves the second part of the invariant. This implies, in particular,  that part 2 of the invariant is true
for nodes at any distance $d$ from node $z$, and hence all nodes of the network declare the same round as {\em red}.

Since there are at most $N$ time intervals used in the procedure and each of them has length $2\nu+4 \in O(\log N)$, the entire procedure takes time $O(N\log N)$. 
\qq
\end{proof}

As we mentioned at the beginning of this section, we want to use our broadcasting algorithm many times, each time starting from a different node. In order to take advantage of the efficiency of broadcasting, which depends on the diameter and 
not on the size of the network,
all nodes need to have a linear upper bound on the diameter of the network. Note, that in order to accomplish one execution of this algorithm,
from one source node, no such upper bound was needed. It becomes needed for multiple broadcasts, as we want to execute each of them in a separate time interval, and thus we need a good estimate of the time of each execution. Recall, that we assume knowledge of a bound $N$ on the size of the network
but no knowledge of any such bound on the diameter. Clearly $N$ is an upper bound on the diameter as well, but may be vastly larger than the diameter.

The following procedure is devoted to obtaining a linear upper bound on the diameter $D$ of a network. It will be executed after the execution of 
Procedure Find Max and Procedure Synchronization. Hence we may assume that the largest of all labels is known to all nodes. Let $z$ be the node
with this label. We may also assume that all nodes declared the same round $r$ as {\em red}. Moreover, each node other than $z$ has its variable $level$ set
to its distance from $z$ (this is done in Procedure Synchronization).

\smallbreak
\noindent
{\bf Procedure Diameter Estimation}

Each node defines consecutive time intervals $J_j=[r+(j-1)N+1,r+jN]$, for positive integers $j$. 
In time interval $J_1$ each node at level $\ell$ beeps in round $\ell$ of this interval.
For $j>1$, if a node at $level$ $\ell$ heard a beep in round $\ell +1$ of interval $J_{j-1}$, then it beeps in round $\ell$ of interval $J_j$.
In the first round $s$ of the form $r+xN+1$ in which $z$ does not hear a beep, it sets $\rho=(s-r-1)/N$. 
Let $\mu$ be the binary representation of the integer $2\rho$ and let $m$ be the size of message $\mu$. 
All nodes execute Algorithm Broadcast with node $z$ as the source and message $\mu$ as the source message.
In this execution the role of round $t$ in which $z$ is woken up is played by round $r+\rho N+1$.
Every node decodes the integer $D^*=2\rho$ as an upper bound on the diameter $D$ of the network.
All nodes declare round $r+\rho N +D^*+6m+2$ as {\em blue}. \alg

\begin{lemma}\label{diamest}
Upon completion of Procedure Diameter Estimation all nodes have the same linear upper bound $D^*$ on the diameter of the graph.
They all declare the same round as {\em blue}, and Procedure Diameter Estimation is completed by this round.
Procedure Diameter Estimation takes time $O(DN)$. 
\end{lemma}

If nodes know a linear upper bound $D^*$ on the diameter of the network, Procedure Find Max can be modified to work faster.
The modifications are detailed below.  
\smallbreak
\noindent
{\bf Procedure Modified Find Max}
 
In Procedure Find Max replace the wakeup round $\tau$ by some round $\xi$, given as input in its call.
Round $\xi+1$ will be called the {\em starting round} of the procedure.
Let $t_j=\xi+j(4D^*+1)$ instead of 
$t_j=\tau+j(4N+1)$, for $j=0,1,\dots,\lambda$. 
Let $I_j=[t_j-2D^*,t_j+2D^*]$ instead of $I_j=[t_j-2N,t_j+2N]$, for $j=1,\dots,\lambda$.
The rest of Procedure Find Max remains unchanged. 
\alg

The proof of the following proposition is the same as that of Lemma \ref{find-max}, using the above modifications.

\begin{proposition}\label{prop1}
At the end of the execution of Procedure Modified Find Max,  there is exactly one active participating node. This node has the largest label among participating nodes.
All nodes know the label of this node. Procedure Modified Find Max takes time $O(D^*\log L)$.
\end{proposition} 

Using Procedure Modified Find Max  we now establish the order between all nodes as follows.
The procedure will be called after the execution of Procedure Find Max, Procedure Synchronization and Procedure Diameter Estimation.
Hence we assume that $z$ is the node with the largest label, found by Procedure Find Max, and that $b$ is the common {\em blue} round found by all nodes
in Procedure Diameter Estimation. All nodes start Procedure Ordering in round $b+1$. Let $x \in O(D^*\log L)$ be an upper bound on the duration of
Procedure Modified Find Max, established in Proposition \ref{prop1}. 

\smallbreak
\noindent
{\bf Procedure Ordering}

$P:=$ the set of all nodes except node $z$

node $z$ assigns itself rank 0

{\bf for} $i:=1$ {\bf to} $N$ {\bf do}\\
\hspace*{1cm}execute procedure Modified Find Max in the time interval\\ 
\hspace*{1cm}$[r+(i-1)x+1,r+ix]$, with the set $P$ of participating nodes;\\ 
\hspace*{1cm}the node found in the current execution of procedure Modified Find Max\\ 
\hspace*{1cm}removes itself from $P$ and assigns itself $rank=i$.
\alg

\begin{lemma}\label{ranks}
Ranks assigned to nodes in the execution of Procedure Ordering define a strictly decreasing function of the node labels.
Procedure Ordering is completed in round $b+Nx$ and takes time $O(ND\log L)$.
 \end{lemma}

\begin{proof}
Since after each execution of Procedure Modified Find Max, the node with the largest label among participating nodes is found,
and this node stops participating in the following executions, the rank assigned to the $j$th largest node is $j-1$.
Since there are $N$ time intervals, each of length $x \in O(D^*\log L)=O(D\log L)$, the time estimate follows.
\qq   
\end{proof}

We are now ready to formulate a gossiping algorithm working for arbitrary networks.
Let $M$ be the upper bound on the size of all input messages, known to all nodes. Let $y \in O(D^*+M)=O(D+M)$ be the upper bound on the duration
of Algorithm Broadcast  established in Theorem \ref{broad-trees}, for the value $D^*$ of the diameter and for the size $M$ of the message.
\smallbreak
\noindent
{\bf Algorithm Gossiping}
\begin{enumerate}
\item
Execute Procedure Find Max; let $z$ be the node with the largest label;
\item
Execute Procedure Synchronization; let $r$ be the {\em red} round found in this procedure;
\item
Execute Procedure Diameter Estimation starting in round $r+1$; let $D^*$ be the upper bound on the diameter of the network found in this procedure;
let $b$ be the {\em blue} round found in Procedure Diameter Estimation;  
\item
Execute Procedure Ordering starting in round $b+1$; 
\item
Let $I_j$ be the time interval $[b+Nx+jy+1,b+Nx+(j+1)y]$;
\item
In time interval $I_j$
 execute Algorithm Broadcast with node of rank $j$, found in Procedure Ordering,
as the source, and the input message of this node as the source message. (For each $j$, the role of the round $t$ when the source is woken up is played by the first round of interval $I_j$.)
\alg
\end{enumerate}

{\bf Remark.} Note that, in the first step of the algorithm, we have to use Procedure Find Max instead of the more efficient Procedure Modified Find Max
because at this point the only estimate on the diameter, known to nodes, is $N-1$. However, since this procedure is executed only once, it does not increase the time complexity of the entire algorithm.

\begin{theorem}
Algorithm Gossiping is a correct gossiping algorithm working in any network of diameter $D$ with at most $N$ nodes in time
$O(N(M+D\log L))$, where $M$ is an upper bound on the size of all input messages, known to all nodes.
\end{theorem}

\begin{proof}
By Lemma \ref{find-max}, Procedure Find Max correctly finds the node $z$ with the largest label. By Lemma \ref{synch},
all nodes compute the same round $r$, and hence start Procedure Diameter Estimation in the same round $r+1$. 
By Lemma \ref{ranks}, there is at most one node with any rank $j \leq N$.
By Lemma \ref{diamest}, $D^*$ is a linear upper bound on the diameter of the network.
Hence, in view of Theorem \ref{broad-trees}, the upper bound $y$ is indeed an upper bound on the
time of execution of  Algorithm Broadcast starting from any source node. This implies that all nodes broadcast their messages in pairwise disjoint time intervals, and hence all broadcasts are correct, in view of Theorem \ref{broad-trees}. This proves the correctness of Algorithm Gossiping.

It remains to estimate the execution time of the algorithm. Procedure Find Max takes time $O(N\log L)$. Procedure Synchronization takes time $O(N\log N)$.
Procedure Diameter Estimation takes time $O(ND)$.
Procedure Ordering takes time $O(ND\log L)$. At most $N$ executions of Algorithm Broadcast in step 6 of Algorithm Gossiping
take time $O(N(D+M))$. Hence Algorithm Gossiping takes time $O(N(M+D\log L))$.
\qq
\end{proof}

We conclude with the following lower bound on the time of gossiping that holds even for the class of trees of bounded diameter.

\begin{proposition}\label{lb}
Assume that all input messages have size $M \geq \gamma \log N$ for some constant  $\gamma >1$. Then there exist bounded diameter trees of size $\Theta(N)$ for which every
gossiping algorithm takes time $\Omega (MN)$. 
\end{proposition}

\section{Conclusion}

We considered two basic communication tasks, broadcasting and gossiping, in the arguably weakest communication model, in which in every round each node has
only the choice to beep or to listen. For the task of broadcasting, we proposed an optimal algorithm  working in time $O(D+m)$ for arbitrary networks of diameter $D$,
where $m$ is the message size.
For the task of gossiping, we presented an algorithm working in time $O(N(M+D\log L))$ for arbitrary networks of diameter $D$ with at most $N$ nodes.
Here $M$ denotes an upper bound on the size of all input messages, known to all nodes. It remains open if this complexity can be improved
in general. Note however, that our gossiping algorithm has optimal time for networks
of diameter bounded by a constant, if the following two assumptions are satisfied: the size of all input messages is the same, it is known to all nodes, and satisfies  $M \geq \gamma \log N$ for some constant  $\gamma >1$, and the size $L$ of the label space is polynomial in $N$. Indeed, in this case $\log L \in O(\log N)$ and $\log N \in O(M)$. Since for bounded $D$, we have $O(N(M+D\log L))=O(N(M+\log L))$, Algorithm Gossiping works in time $O(NM)$ in this case, which matches the lower bound $\Omega (MN)$, shown in
Proposition \ref{lb}. 

The two above
assumptions do not seem to be overly strong. Indeed, in most applications, we want messages exchanged by gossiping nodes to be large enough to hold at least the node's
label and some other useful information, which justifies the assumption $M \geq \gamma \log N$ for some constant  $\gamma >1$.
On the other hand,  labels of nodes are often assumed to be polynomial in the size of the network, as there is usually no
need of larger labels. Notice, moreover, that if these assumptions
are satisfied, the complexity of our gossiping algorithm is $O(NMD)$, i.e., it exceeds the lower bound $\Omega (MN)$ only by a factor of $D$, for any value
of the diameter $D$. Thus, our gossiping algorithm is not only optimal for networks of bounded diameter,  but it is close to optimal for ``shallow'' networks, e.g., 
networks whose diameter is logarithmic in their size, such as the hypercube.

\bibliographystyle{plain}


\pagebreak

\begin{center}
{\bf \Large APPENDIX}
\end{center}

\vspace*{0.7cm}
\noindent
{\bf Proof of Lemma \ref{find-max}}

Let $\rho$ be the earliest round in which some node is  woken up by the adversary. (There may be several nodes woken up in round $\rho$).
Since each node beeps in the round following its wakeup (either by the adversary or by a beep), all nodes are woken up until round $\rho+N$.

We will use the following claim.

{\bf Claim.} If a node $v$ heard a beep in stage $j$,
then there exists a node $w$ active in stage $j$, with the binary representation $\ell=(c_1,\dots,c_{\lambda})$ of its label, such that
$c_j=1$.

In order to prove the claim, first note that stage $j$ of node $v$ starts after round $\rho +N$. 
For any node $u$ and any non-negative integer $i$, let $t_i(u)$ denote the round $t_i$ computed by node $u$ (relative to its wakeup round).
Since non-active nodes beep in some round $s>\rho +N$
only if they heard a beep in round $s-1$, it follows that if $v$ heard a beep in some round $t$ of its stage $j$, then an active node $w$
must have beeped in some round $t-N <t' \leq t$, such that no node beeped in round $t'-1$. According to the procedure, the only reason for such a beep is that, for some $k$,  $t'=t_k(w)$, node $w$ is active in its stage $k$, and that $c_k=1$, where $\ell=(c_1,\dots,c_{\lambda})$ is the binary representation of the label of $w$. Suppose that $k<j$. This implies $t_j(v)-3N \leq t-N <t_k(w)\leq t_{j-1}(w)=t_j(w)-(4N+1)$, hence $t_j(v)+N+1 <t_j(w)$.
Hence the wakeup rounds of nodes $v$ and $w$ differ by more than $N$, which is a contradiction. 
Thus $k$ cannot be smaller than $j$. For similar reasons, $k$ cannot be larger than $j$. This leaves the only possibility of $k=j$, which proves the claim.

Next, we prove that the node $z$ with the largest label among participating nodes remains active till the end of its stage $\lambda$.
Let  $(d_1,\dots,d_{\lambda})$ be the binary representation of this label.
Suppose, for contradiction, that $z$ becomes passive in some stage $j \leq \lambda$. According to the procedure, this implies that it heard a beep in stage $j$
and that $d_j=0$. In view of the claim, there is a node $w$ active in stage $j$, with the binary representation $(c_1,\dots,c_{\lambda})$ of its label, such that
$c_j=1$. Consider any index $k<j$. If $c_k=0$ then also $d_k=0$. Otherwise, since node $z$ is active in stage $k$, it beeps in stage $k$
and hence node $w$ would become passive in stage $k$, which contradicts the fact that it became passive only in stage $j$. This proves that the sequence
$(d_1,\dots,d_{\lambda})$ is lexicographically smaller than the sequence $(c_1,\dots,c_{\lambda})$, which contradicts the assumption that
$z$ has the largest label among participating nodes. 

Further, we prove that no node other than the node $z$ with the largest label among participating nodes remains active till the end of its stage $\lambda$.
Let $w$ be such a node with the binary representation $(c_1,\dots,c_{\lambda})$ of its label, and let $s$ be the first index where $d_s \neq c_s$.
Since $z$ is active in stage $s$ and $d_s=1$, node $z$ beeps in stage $s$, and hence $w$ hears it and becomes passive (at the latest) in stage~$s$.

It follows that $z$ is the only participating node that is active at the end of the execution of procedure Find Max. Node $z$ knows that it remained active at the end, so it knows that it has the largest label. Every other node $w$ (it is passive at the end of the procedure execution) deduces the binary representation $(d_1,\dots,d_{\lambda})$
of the label of $z$ as follows: $d_i=1$, if and only if $w$ heard a beep in stage~$i$. Indeed, if $d_i=1$ then it beeped in stage $i$, because it was active in this stage, and hence $w$ heard a beep at most $N$ rounds later, hence still in its stage $i$. If  $d_i=0$, then no node beeped in stage $i$ because all nodes that have 1
as the $i$th term of the binary representation of their label must be already passive in stage $i$, as this representation is lexicographically smaller than $(d_1,\dots,d_{\lambda})$. 

Since each time interval has length $O(N)$ and there are $\lambda=\lceil \log L \rceil$ intervals, the whole procedure takes time $O(N\log L)$ since the wakeup of the earliest node.

\vspace*{0.7cm}
\noindent
{\bf Proof of Lemma \ref{diamest}}

Let $\delta$ be the largest distance of any node from $z$.
For every node other than $z$, the value of the variable $level$ is its distance from $z$.
Hence in time interval $J_j$ exactly nodes at distance  at most $\delta +1-j$ from $z$ beep. It follows that $s=r+\delta N +1$ is the first round
of the form $r+xN+1$ in which
$z$ does not hear a beep. Consequently $\rho=\delta$. By the correctness of  Algorithm Broadcast, all nodes correctly decode the integer
$D^*=2\rho$. Since $\rho$ is the largest distance of any node from $z$, the diameter $D$ of the network satisfies inequalities $\rho \leq D \leq 2\rho$.
Thus $D^*$ is a linear upper bound on the diameter of the network. 

After decoding the integer $D^*$, all nodes know $N$, $r$, $\rho$ and $m$. Hence the round declared as {\em blue} is the same for all nodes.  
It was follows from the proof of Theorem \ref{broad-trees} that 
Algorithm Broadcast takes time at most $D+6m+2 \leq D^*+6m+2$.  Since $z$ starts broadcasting in round $r+\rho N+1$, by round
{\em blue} the procedure is completed. It takes time $O(\rho N +D^*+m)=O(\rho N +D^*+\log \rho)= O(DN+D+\log N)=O(DN)$.

\vspace*{0.7cm}
\noindent
{\bf Proof of Proposition \ref{lb}}

Consider the star $S$ with $\lfloor N/2 \rfloor$ leaves, i.e., a tree with one node $v$ of degree $\lfloor N/2 \rfloor$ and with $\lfloor N/2 \rfloor$ nodes of degree 1.
$S$ is a tree of diameter 2.
Let $w$ be any leaf. Consider any algorithm $\cal A$ accomplishing gossiping in $S$ in time $t$. The leaf $w$ can obtain information only from node $v$.
In time $t$ node $v$ can transmit $2^t$ sequences with terms in the set $\{b,s\}$, where $b$ denotes a beep and $s$ denotes silence. 

Consider the family $Z$ of possible sets of input messages initially held by the $\lfloor N/2 \rfloor$ nodes of $S$ other than $w$, assuming that each node has a different message of size $M$.
 If $|Z|>2^t$, then
node $v$ executing algorithm $\cal A$ must generate  the same sequence with terms in the set $\{b,s\}$ for two distinct sets of messages initially held by the $\lfloor N/2 \rfloor$ nodes of $S$ other than $w$, and consequently $w$ cannot correctly deduce the set of messages held by other nodes of $S$. This implies that $|Z|\leq 2^t$. 

Since there are $2^M$ possible messages of size $M$, the set $Z$ has size  ${2^M}\choose {\lfloor N/2 \rfloor}$.
For sufficiently large $N$ we have  $|Z|={{2^M}\choose {\lfloor N/2 \rfloor}}\geq (2^M/N)^{\lfloor N/2 \rfloor}$, hence\\
 $$t \geq \log |Z|= \log {{2^M}\choose {\lfloor N/2 \rfloor}} \geq \lfloor N/2 \rfloor(M-\log N) \geq (1-1/\gamma )M\lfloor N/2 \rfloor \in \Omega(MN),$$
in view of $M \geq \gamma \log N$.

\end{document}